\documentclass[11pt]{article}

\usepackage{amssymb,amsthm,amsmath}
\usepackage{graphicx}
\usepackage[small]{caption}
\usepackage{subcaption}
\usepackage{epsfig}
\usepackage{amsfonts}
\usepackage{complexity}

\newcommand{\later}[1]{}
\newcommand{\old}[1]{}

\usepackage{amsfonts}

\usepackage[utf8]{inputenc}
\usepackage{fullpage}
\usepackage{framed}

\usepackage{enumerate}
\usepackage{url}
\usepackage{hyperref}

\newtheorem{theorem}{Theorem}
\newtheorem{lemma}[theorem]{Lemma}

\newtheorem{corollary}[theorem]{Corollary}

\newcommand{\ie}{{i.e.}}
\newcommand{\eg}{{e.g.}}

\newcommand{\alg}{\textsf{ALG}}
\newcommand{\opt}{\textsf{OPT}}

\newcommand{\ZZ}{\mathbb{Z}} 
\newcommand{\eps}{\varepsilon}

\def\A{\mathcal A}

\def\E{{\rm E}}

\title{General Position Subset Selection in Line Arrangements}

\author{%
Adrian Dumitrescu\footnote{Algoresearch L.L.C., Milwaukee, WI 53217, USA,
  E-mail: \texttt{ad.dumitrescu@algoresearch.org}}
}

\begin{document}

\maketitle

\begin{abstract}
  Given a set of points in the plane, the \textsc{General Position Subset Selection} problem
  is that of finding a maximum-size subset of points in general position,
  \ie, with no three points collinear. The problem is known to be $\NP$-complete and $\APX$-hard,
  and the best approximation ratio known is $\Omega\left(\opt^{-1/2}\right) =\Omega(n^{-1/2})$.
  Here we obtain better approximations in three specials cases:

  (I) A constant factor approximation for the case where the input set
  consists of lattice points and is \emph{dense}, which means that the ratio between the
  maximum and the minimum distance in $P$ is of the order of $\Theta(\sqrt{n})$.

  (II) An $\Omega\left((\log{n})^{-1/2}\right)$-approximation for the case where the input set
  is the set of vertices of a \emph{generic} $n$-line arrangement, \ie, one with $\Omega(n^2)$ vertices. 

  (III) An $\Omega\left((\log{n})^{-1/2}\right)$-approximation for the case where the input set
  has at most $O(\sqrt{n})$ points collinear and can be covered by $O(\sqrt{n})$ lines.
  
  The scenario in (I) is a special case of that in (II).
  Our approximations rely on probabilistic methods and results from incidence geometry.
\end{abstract}

\section{Introduction} \label{sec:intro}

A set of points in the plane is said to be in \emph{general position} if no $3$ points are collinear.
The problem of selecting a large subset of points from the $k \times k$ grid of integer points
with no three points collinear (\ie, in general position) goes back more than $100$ years,
see for example~\cite{Dud1917} and~\cite[Ch.~10]{BMP05}.
In particular, it is not known if one can always select $2k$ points from the $k \times k$ grid
so that no $3$ points are collinear~\cite{Dud1917}.

A key quantity in the process of selecting a large subset from an input set of points
is the number of \emph{collinear triples} in the set. 
Payne and Wood~\cite{PW13} obtained the following upper bound on this number.
The proof relies on the classical points and lines incidence bound due to
Szemer\'edi and Trotter~\cite{ST83}; see also~\cite[Chap.~10]{PA95} and~\cite{Sze97}
for a modern approach of this topic.
(The special case $\ell= \sqrt{n}$ can also be found in \cite{Lef08} or
in~\cite[p.~313]{TV06}.)

\begin{lemma} {\rm \cite{PW13}.} \label{lem:T}
  Let $P$ be a set of $n$ points in the plane with at most $\ell$ collinear.
  Then the number of collinear triples in $P$ is $T = O(n^2 \log{\ell} + \ell^2 n)$.
\end{lemma}

By applying the above lemma together with a  lower bound on the independence number of
a hypergraph due to Spencer~\cite{Spe72}, Payne and Wood~\cite{PW13} obtained the
following result:

\begin{theorem} {\rm \cite{PW13}.} \label{thm:payne-wood}
  Let $P$ be a set of $n$ points in the plane with at most $\ell$ collinear.
  Then $P$ contains a subset of $\Omega\left(n/\sqrt{n \log{\ell} + \ell^2}\right)$ points
  in general position. In particular, if $\ell =O(\sqrt{n})$, then $P$ contains
  a subset of $\Omega\left((n/\log{\ell})^{1/2}\right)$ points in general position.
\end{theorem}

A cursory examination of the above lower bound shows that if the input set has a large
subset of points in general position (of, say, linear, or nearly linear size), then
the above guarantee is nearly $\sqrt{n} \ll n$, \ie, roughly a factor of $\sqrt{n}$
smaller than the truth. 

  Given a set $P$ of points in the plane, the \textsc{General Position Subset Selection} problem
  (\textsc{GPSS}, for short), is that of finding a maximum-size subset of points in general position.
  The current best approximation for the problem is a greedy algorithm due to Cao~\cite[Chap.~3]{Cao12},
  achieving a ratio of $\Omega\left(\opt^{-1/2}\right) =\Omega(n^{-1/2})$;
  here $\opt$ is the value of an optimum solution to \textsc{GPSS}.
  See also~\cite[Chap.~9]{Epp18} and~\cite{FKNN17} for further aspects of this problem.

\paragraph{Preliminaries.}
For a set $P$ of $n$ points in the plane, consider the ratio (sometimes also called \emph{spread})
\[ D(P)= \frac{\max\{|ab| \colon a,b \in P, a \neq b\}}{\min\{|ab| \colon a,b \in P, a \neq b\}}, \]
where $|ab|$ is the Euclidean distance between points $a$ and $b$.
We may assume without loss of generality that $\min\{|ab| \colon a,b \in P, a \neq b\}=1$. 
In this case $D(P) = \max\{|ab| \colon a,b \in P, a \neq b\}$ is the diameter of $P$.
A standard disk packing argument shows that if $|P|=n$, then $D(P) \geq \alpha_0 \, n^{1/2}$, where
\begin{equation}\label{alfa0}
\alpha_0:= 2^{1/2} 3^{1/4} \pi^{-1/2} \approx 1.05,
\end{equation}
provided that $n$ is large enough; see~\cite[Prop.~4.10]{Va92}.
On the other hand, a $\sqrt{n} \times \sqrt{n}$ section of the integer lattice shows that
this bound is tight up to a constant factor.
An $n$-element point set $P$ satisfying the condition $D(P) \leq \alpha \, n^{1/2}$,
for some constant $\alpha\ge\alpha_0$, is said to be $\alpha$-\emph{dense}; see
for instance~\cite{KT20} and Fig.~\ref{fig:grid} for an illustration.

According to a classical result of Beck~\cite{Be83}, if $P$ is a set of $n$ points in the plane
where at most $\ell$ points of $P$ are collinear, then $P$ determines at least $\Omega(n(n-\ell))$
distinct lines. In particular, if $\ell \leq cn$, where $c<1$ is a constant, then
 $P$ determines $\Omega(n^2)$ distinct lines.
By duality, see \eg, \cite{Mat02}, for a set $L$ of $n$ lines in the plane,
where at most $\ell$ lines of $L$ are concurrent, the number of vertices of the corresponding
arrangement is $\Omega(n(n-\ell))$. We say that an arrangement of $n$ lines is \emph{generic}
if it has $\Omega(n^2)$ vertices.

Next, we introduce the setup.
Given a set $P$ of points in the plane, let $\ell(P)$ denote the largest number of points
in $P$ that is incident to a line determined by $P$. We have $2 \leq \ell(P) \leq |P|$,
with  $\ell(P)=2$ if and only if $P$ is in general position and
$\ell(P)=|P|$ if and only if $P$ is collinear.

Given a set $P$ of points in the plane, the \emph{line-cover} of $P$,
denoted $\kappa(P)$, is the minimum number of lines needed to cover all points in $P$.
Computing  $\kappa(P)$ is $\APX$-hard, and the best approximation ratio known is
$O(\log{n})$~\cite{DJ15}.

As in~\cite{WS11}, here we follow the convention that the approximation ratio of an algorithm
for a maximization problem is less than $1$. Throughout this paper all logarithms are in base~$2$.

\paragraph{Our results.}
First, we obtain a constant factor approximation for the \textsc{General Position Subset Selection}
problem in dense lattice point sets; its approximation ratio depends on the spread ratio of the
input. 

\begin{theorem} \label{thm:dense}
  Given an $\alpha$-dense set of $n$ lattice points, a constant factor approximation
  for the \textsc{General Position Subset Selection} problem can be computed in polynomial time.
  If $c=c(\alpha)$ denotes its approximation ratio, then $c(\alpha) = \Omega\left(\alpha^{-2}\right)$.
\end{theorem}

Second, we obtain an $\Omega((\log{n})^{-1/2})$-approximation for the special case where the input set
  is the set of vertices of an $n$-line arrangement under the relatively mild genericity assumption.

\begin{theorem} \label{thm:arrangement}
  Given $n$ lines in the plane that make a generic line arrangement, 
  an $\Omega\left((\log{n})^{-1/2}\right)$-approximation
  for the \textsc{General Position Subset Selection} problem for the set of vertices
  of the corresponding line arrangement, can be computed by a randomized algorithm in expected
  polynomial time.
\end{theorem}

Essentially the same proof (of this theorem in Section~\ref{sec:arrangement}) yields the following:

\begin{corollary} \label{cor:arrangement}
  Given $n$ lines in the plane, let $V$ denote the vertex set of the corresponding line arrangement.
  If $V' \subseteq V$ is a subset of vertices with $|V'| = \Omega(n^2)$, 
  an $\Omega\left((\log{n})^{-1/2}\right)$-approximation
  for the \textsc{General Position Subset Selection} problem
  in $V'$ can be computed by a randomized algorithm in expected
  polynomial time.
\end{corollary}

  It is worth noting that the scenario in Theorem~\ref{thm:dense} is a special case
  of the scenario in Theorem~\ref{thm:arrangement}. Indeed any finite set $P$ of lattice points in the
  integer grid is a subset of the vertices of an arrangement $\A^+$ of axis-parallel lines,
  induced by $H$ and $V$, the sets of horizontal and vertical lines incident to points in $P$,
  respectively. Moreover, if $P$ is a dense set, then $\A^+$ is generic, since
  \[ |H|,|V|  = \Omega(|P|/(\alpha \sqrt{n})) = \Omega(\sqrt{n}), \]
  thus $\A^+$ has $|H| \cdot |V| = \Omega(n)$ vertices. 
  Observe that \emph{not} all vertices in the arrangement $\A^+$ are necessarily in $P$.

  Third, we obtain an $\Omega((\log{n})^{-1/2})$-approximation for the special case where the input set
  has at most $O(\sqrt{n})$ points collinear and can be covered by $O(\sqrt{n})$ lines.

\begin{theorem} \label{thm:special}
  Given a set $P$ of $n$ points in the plane with $\ell(P) =O(\sqrt{n})$ and
  $\kappa(P) =O(\sqrt{n})$,  an $\Omega\left((\log{n})^{-1/2}\right)$-approximation
  for the \textsc{General Position Subset Selection} problem can be computed by
  a randomized algorithm in expected polynomial time.
\end{theorem}

\section{Subset Selection in Dense Lattice Point Sets} \label{sec:dense}

In this section we prove Theorem~\ref{thm:dense}. First, we introduce the setup.

\smallskip
\begin{figure}[htbp]
\centering
\includegraphics[scale=0.9]{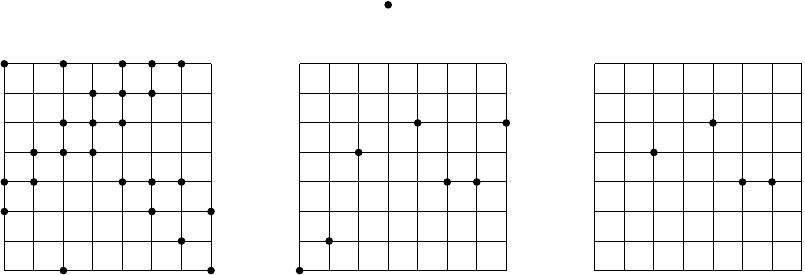}
\caption {Left: a $2$-dense set of $n=25$ points in the $8 \times 8$ grid ($m=8$). 
  Center: $V_0$; $|V_0|=m=8$ and $p=11$; points in $V_i$ may lie outside the grid,
  \eg, $(3,9)$ lies $2$ units above it.
  Right: The approximation algorithm returns $P \cap V_0$ (or $P \cap V_3 $);
  here $|P \cap V_0|=|P \cap V_3|=4$.}
\label{fig:grid}
\end{figure}

For a positive integer $m$, the $m \times m$ grid is the set of points in the plane
$G_m= \{(x, y) \colon x,y \in \{0,1,\ldots,m-1\}$. We have $|G_m|=m^2$.
Following~\cite{Wood04}, let $k(m)$ denote the minimum number of colors in a coloring
of the points in $G_m$ such that no three collinear points are monochromatic.
Since  no three points in a single row or column can receive the same color,
we have $k(m) \geq m/2$. It was shown by Wood~\cite{Wood04} that for any $\eps>0$,
$k(m) \leq (2+\eps)m$ for every $m \geq M(\eps)$. Among others, this is based on the fact that
for any $\eps>0$, there exists a prime between $m$ and $(1+\eps)m$, for every
$m \geq M(\eps)$. For instance, a nonasymptotic result, more convenient for small $m$,
is that there exists a prime between $m$ and $6m/5$, for every $m \geq 25$~\cite{Na52}.
(And by the well-known Bertrand-Chebyshev theorem,  there exists a prime
between $m$ and $2m$, for every $m$.)

We briefly recall Wood's argument~\cite{Wood04} relying on a construction of Erd\H{o}s~\cite{Erd51}.
Let $p \geq m$ be a prime; in practice we may pick the smallest prime at least $m$ or
a prime at least $m$ that is close to $m$. 
For any integer $i$, let
\[ V_i = \{(x,(x^2 \bmod p) +i) \colon x=0,1,\ldots,m-1\} . \]
Note that $|V_i|=m$, and that $V_i$ may be \emph{not} contained in $G_m$.
Refer to Fig.~\ref{fig:grid}.
A Vandermonde determinant calculation shows that $V_i$ is
in general position for every $i \in \ZZ$. Moreover, $V_i \cap V_j = \emptyset$ whenever $i \neq j$.
Each point $(x,y) \in G_m$ is in $V_i$ where $i=y-(x^2 \bmod p)$. 
Since
\[1-p = 0 -(p-1) \leq y-(x^2 \bmod p) \leq m-1, \]
the union of the $m+p-1$ sets $V_i$
over $i=1-p,\ldots,m-1$ consists of $m+p-1$ color classes where each class is in general position.
That is, $G_m$ can be covered by $m+p-1$ $m$-element point sets in general position. 
As mentioned in the previous paragraph, $m+p-1 \leq (2+\eps)m$ for every $m \geq M(\eps)$,
or $m+p-1 \leq 11m/5$ for every $m \geq 25$, as needed.

\paragraph{Proof of  Theorem~\ref{thm:dense}.} Since $P$ is $\alpha$-dense, we may assume that
$P \subseteq G_m$, where $m = \lceil \alpha \sqrt{n} \rceil$, after a suitable translation
by an integer vector. The algorithm chooses a prime $p \geq m$ close to $m$.
By the Prime Number Theorem, this involves \emph{primality testing}~\cite[Ch.~14.6]{MR95}
of no more than $O(m^{0.525})=O(n^{0.263})$ odd integers in the interval $[m,m+2m^{0.525}]$~\cite{BHP01}. 
As argued above, $G_m$ is covered by $m+p-1 \leq (2+\eps)m$ $m$-element point sets $V_i$
in general position. By monotonicity
(\ie, every subset of a set in general position is likewise in general position),
the sets $P \cap V_i$ are in general position, and the algorithm outputs the largest one,
with at least $\frac{n}{(2+\eps)m}$ elements. 

Since each of the $m$ rows of $G_m$ contains at most two elements of $P$, we have
$\opt \leq 2m$. Consequently, the ratio $\alg/\opt$ is bounded from below as
\[ \frac{\alg}{\opt} \geq \frac{n}{2 \cdot (2+\eps) m^2} =
\Omega\left( \frac{1}{\alpha^2} \right). \]
Clearly the resulting algorithm runs in polynomial time; we give a few details below.

Choosing the prime $p$ and generating the sets $V_i$ takes $O(n)$ time.
Selecting one $V_i$ with a maximal size $|P \cap V_i|$ can be done withing the
same time, and so the overall running time is $O(n)$. 
\qed

\paragraph{Remark.} Using a deterministic algorithm for primality testing makes the approximation
algorithm deterministic as well.

\section{Subset Selection in Generic Line Arrangements} \label{sec:arrangement}

In this section we prove Theorem~\ref{thm:arrangement}. Recall that an arrangement
of $n$ lines is generic if it has $\Omega(n^2)$ vertices. See Fig~\ref{fig:generic}.

\smallskip
\begin{figure}[htbp]
\centering
\includegraphics[scale=0.6]{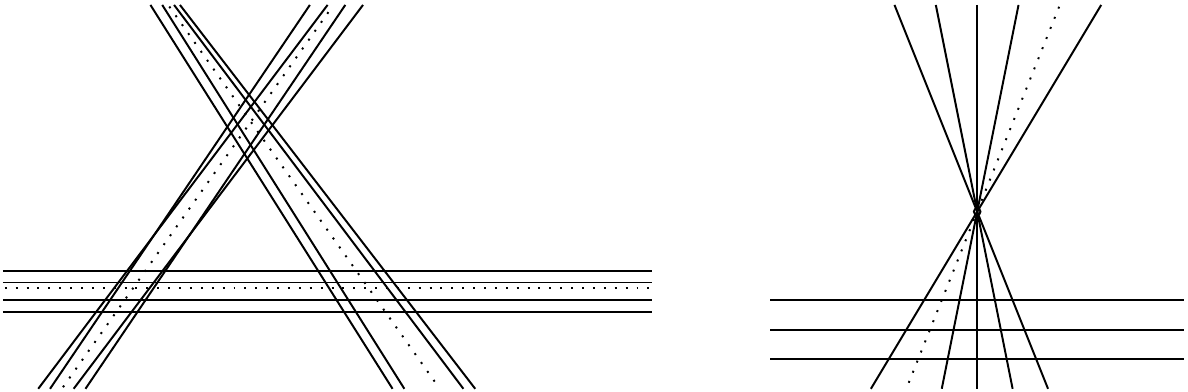}
\caption {Left: a generic line arrangement consisting of three bundles of $n/3$ nearly parallel
  lines each. 
  Right: a non-generic line arrangement consisting of $3$ parallel lines and $n-3$ concurrent lines.}
\label{fig:generic}
\end{figure}

\begin{lemma} \label{lem:ell}
  Let $V$ be the set of vertices of an $n$ line arrangement $\A= \A(L)$, where
  $L=\{\ell_1,\ldots,\ell_n\}$. Then $\ell(V) \leq n-1$.  
\end{lemma}
\begin{proof}
  We distinguish two cases. Let $V' \subseteq V$ be a set of collinear vertices, say on a line $h$.
  If $h \in L$, then obviously $|V'| \leq n-1$.
  If $h \notin L$, scan $h$, say, from left to right, and observe that each vertex in $V'$ uses up
  at least two ``new'' lines from $L$, thus $|V'| \leq n/2$, see Fig.~\ref{fig:arr}. 
  \qed
\end{proof}

\smallskip
\begin{figure}[htbp]
\centering
\includegraphics[scale=0.8]{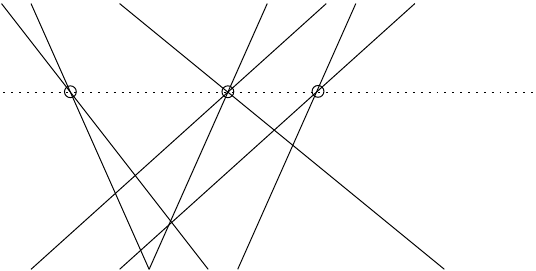}
\caption {Collinear vertices in a line arrangement --incident to an induced line.}
\label{fig:arr}
\end{figure}

Let $V$ be the set of vertices of a generic $n$ line arrangement $\A= \A(L)$.
Let $N=|V| \leq {n \choose 2}$; we also have $N \geq c n^2$ by the assumption.
By Lemma~\ref{lem:ell}, $\ell:= \ell(V) \leq n-1$, thus by Lemma~\ref{lem:T}
the number of collinear triples is bounded as:
\begin{align} \label{eq:T}
  T &= O(N^2 \log{\ell} + \ell^2 N) =O(N^2 \log{n} + n^2 N).
\end{align}

\paragraph{Proof of  Theorem~\ref{thm:arrangement}.}
The subset selection algorithm consists of two steps: (i)~random sampling from $V$,
and (ii)~application of the deletion method. See, \eg, \cite[Chap.~3]{AS16}; and a similar
application in~\cite{Zha93}.
In the first step, a random subset $X \subset V$ is chosen by selecting points
independently with probability $p=k/N$, for a suitable $k$ to be determined.
Note that $\E[|X|]= N \cdot k/N =k$. 

To select a subset in general position, one needs to avoid one type
of obstacles, collinear triples. One obstacle can be eliminated by
deleting one point from the subset in the second step. In particular,
it suffices to choose $k$ so that 
\[ \E[T p^3] \leq k/2, \]
since this implies that the expected number of remaining points is
at least $k-k/2=k/2$. Let $c_1$ denote the constant hidden in~\eqref{eq:T}. 
Using the previous upper bound on $T$, and recalling that $N \geq c n^2$,
it suffices to choose $k$ so that
\begin{equation} \label{eq:zhang}
  2 c_1 k^2 (N \log{n} + n^2) \leq N^2.
\end{equation}
Setting $k =c' \frac{n}{\sqrt{\log{n}}}$ for a sufficiently small $c'>0$ 
satisfies the inequality, proving the lower bound on the size of the output.
In addition, this setting is valid for ensuring that $p= k/N < 1$.  

In summary, the expected number of points in general position found by the algorithm is
$\Omega \left( \frac{n}{\sqrt{\log{n}}} \right)$.
To analyze the approximation ratio it suffices to notice that $\opt \leq 2n$.
Indeed, for each line $\ell_i \in L$, $i=1,\ldots,n$, at most two vertices in $V$ can
appear in any solution, thus $\opt \leq 2n$.
It follows that the approximation ratio is 
\[ \frac{\alg}{\opt} = \frac{1}{2n} \cdot \Omega\left(\frac{n}{\sqrt{\log{n}}} \right)
= \Omega \left( \frac{1}{\sqrt{\log{n}}} \right), \]
as claimed. 

Since repeated random samplings are independent, 
the probability of error for the resulting randomized algorithm of not finding the required
number of points can be made arbitrarily small by repetition --in a standard fashion.
Clearly the resulting algorithm runs in expected polynomial time; we give a few details below.

Computing the set of vertices $V$ and performing the random sampling takes $O(n^2)$ time.
The expected sample size is $O(k) =O(n)$. Using the point-line duality, the algorithm
computes the line arrangement dual to the point sample in $O(n^2)$ time~\cite[Ch.~8]{BCKO08}.
It repeatedly removes lines passing through vertices incident to at least three lines,
each in $O(n)$ time, and updates the line arrangement. This step is repeated until
the arrangement is simple and thus, the points left in the sample are in general position.
The expected running time is $O(n^2)$. 
\qed

\paragraph{Proof of Corollary~\ref{cor:arrangement}.}
We perform random sampling from $V'$, let $N=|V'|$,
and proceed as in the proof of Theorem~\ref{thm:arrangement}.
Inequality~\eqref{eq:zhang} is again satisfied due to the assumption $|V'| =\Omega(n^2)$.
\qed

\section{Another Application} \label{sec:application}

In this section we restrict ourselves to general position subset selection in
grid-like sets, that is, sets satisfying $\ell(P) =O(\sqrt{n})$ and $\kappa(P) =O(\sqrt{n})$.
The approximation is obtained with an approach similar to that in the proof of
Theorem~\ref{sec:arrangement}. 

\paragraph{Proof of  Theorem~\ref{thm:special}.}
  Let $\ell=\ell(P)$, and $\kappa=\kappa(P)$.
  By Lemma~\ref{lem:T} and the first assumption, the number of collinear triples in $P$ is
\begin{equation} \label{eq:T2}
  T = O(n^2 \log{\ell}) = O(n^2 \log{n}).
\end{equation}
  Since an optimal solution can contain at most two points of $P$
  from each line in an optimal line cover of $P$, we have $\opt \leq 2 \kappa = O(\sqrt{n})$. 

The subset selection algorithm consists of two steps: (i)~random sampling from $P$,
and (ii)~application of the deletion method. 
In the first step, a random subset $X \subset P$ is chosen by selecting points
independently with probability $p=k/n$, for a suitable $k$ to be determined.
Note that $\E[|X|]= n \cdot k/n =k$. 

As in the proof of Theorem~\ref{thm:arrangement}
it suffices to choose $k$ so that 
\[ \E[T p^3] \leq k/2. \]
Let $c_1$ denote the constant hidden in~\eqref{eq:T2}. 
Using the upper bound on $T$ in~\eqref{eq:T2}, it suffices to choose $k$ so that
\begin{equation} \label{eq:zhang2}
  2 c_1 k^2 \log{n} \leq n.
\end{equation}
Setting $k =c' \sqrt{\frac{n}{\log{n}}}$ for a sufficiently small $c'>0$ 
satisfies the inequality, proving the lower bound on the size of the output.
In addition, this setting is valid for ensuring that $p= k/n < 1$.  

The expected number of points in general position found by the algorithm is
$\Omega \left( \sqrt{\frac{n}{\log{n}}} \right)$.
Recalling that $\opt \leq 2 \kappa = O(\sqrt{n})$, 
it follows that the approximation ratio is 
\[ \frac{\alg}{\opt} \geq \frac{1}{2\kappa} \cdot \Omega \left( \sqrt{\frac{n}{\log{n}}} \right)
= \Omega \left( \frac{1}{\sqrt{\log{n}}} \right), \]
as claimed. 
\qed

\section{Conclusion} \label{sec:conclusion}

We conclude with some problems for further investigation:

\begin{enumerate} \itemsep 3pt
  
\item Can the approximation factor in Theorem~\ref{thm:arrangement} be improved?
 Perhaps using ideas from~\cite[Sec.~2]{BCDL24}?

\item Is there a constant factor approximation for the problem of finding
  a largest subset in general position among the vertices of an $n$-line arrangement?

\item Can a constant factor approximation for \textsc{GPSS} on point sets 
satisfying $\ell(P) =O(\sqrt{n})$ and $\kappa(P) =O(\sqrt{n})$ be obtained?
Note that one can select $\Theta(\sqrt{n})$ points in general position
  from the $\sqrt{n} \times \sqrt{n}$ grid of $n$ points, see~\eg, \cite[Ch.~3.3]{AS16}
  or~\cite[Ch.~10]{BMP05}.

\item Can better approximation factors for the general position subset selection problem
  be obtained in other interesting scenarios?
  
\end{enumerate}

\end{document}